%% file: sample-edbt2021.tex
\documentclass[sigconf,edbt]{acmart-edbt2021}

\def\BibTeX{{\rm B\kern-.05em{\sc i\kern-.025em b}\kern-.08em
    T\kern-.1667em\lower.7ex\hbox{E}\kern-.125emX}}

\usepackage{booktabs} 
\usepackage{multirow}
\usepackage{siunitx}

\usepackage{amsmath,amsthm}
\usepackage[ruled,vlined,linesnumbered]{algorithm2e}

\usepackage{graphicx}
\usepackage{float}
\usepackage{subcaption}
\usepackage{makecell}

\newtheorem{fact}[theorem]{Fact}

\setcopyright{rightsretained}

\acmDOI{}

\acmISBN{978-3-89318-084-4}

\acmConference[EDBT 2021]{24th International Conference on Extending Database Technology (EDBT)}{March 23-26, 2021}{Nicosia, Cyprus} 
\acmYear{2021}

\begin{document}
\title{Efficient Maintenance of Distance Labelling for Incremental Updates in Large Dynamic Graphs}
  
\author{Muhammad Farhan}
\affiliation{%
  \institution{Australian National University}
  \streetaddress{P.O. Box 1212}
  \city{Canberra}
  \state{Australia}
  \postcode{43017-6221}
}
\email{muhammad.farhan@anu.edu.au}

\author{Qing Wang}
\orcid{0000-0002-1825-0097}
\affiliation{%
  \institution{Australian National University}
  \streetaddress{1 Th{\o}rv{\"a}ld Circle}
  \city{Canberra}
  \country{Australia}
}
\email{qing.wang@anu.edu.au}

\begin{abstract}
Finding the shortest path distance between an arbitrary pair of vertices is a fundamental problem in graph theory. A tremendous amount of research has been successfully attempted on this problem, most of which is limited to static graphs. Due to the dynamic nature of real-world networks, there is a pressing need to address this problem for dynamic networks undergoing changes. In this paper, we propose an \emph{online incremental} method to efficiently answer distance queries over very large dynamic graphs. Our proposed method incorporates incremental update operations, i.e. edge and vertex additions, into a highly scalable framework of answering distance queries. We theoretically prove the correctness of our method and the preservation of labelling minimality. We have also conducted extensive experiments on 12 large real-world networks to empirically verify the efficiency, scalability, and robustness of our method. 
\end{abstract}

%
%



\maketitle

\input{section_Introduction}
\input{section_RelatedWork}
\input{section_Preliminaries}

\input{section_IncrementalUpdates}

\input{section_TheoreticalResults}

\input{section_Experiments}

\input{section_Conclusion}

\bibliographystyle{ACM-Reference-Format}
\bibliography{sample-base}

%
% ...

\end{document}

%% file: section_Introduction.tex
\section{INTRODUCTION}\label{section:intro}

Given a very large graph with billions of vertices and edges, how efficiently can we find the shortest path distance between any two vertices? If such a graph is dynamically changing over time (e.g. inserting edges or vertices), how can we not only efficiently but also accurately find the shortest path distance between any two vertices? These questions are intimately related to distance queries on dynamic graphs. As one of the most fundamental operations on graphs, distance queries have a wide range of real-world applications that operate on increasingly large dynamic graphs, such as context-aware search in web graphs \cite{ukkonen2008searching}, social network analysis in social networks \cite{vieira2007efficient,backstrom2006group}, management of resources in computer networks \cite{boccaletti2006complex}, and so on. Many of these applications use distance queries as a building block to realise more complicated tasks, and require distance queries to be answered instantly, e.g. in the order of milliseconds. 

\begin{figure}[t!]
\centering
\includegraphics[width=0.4\textwidth]{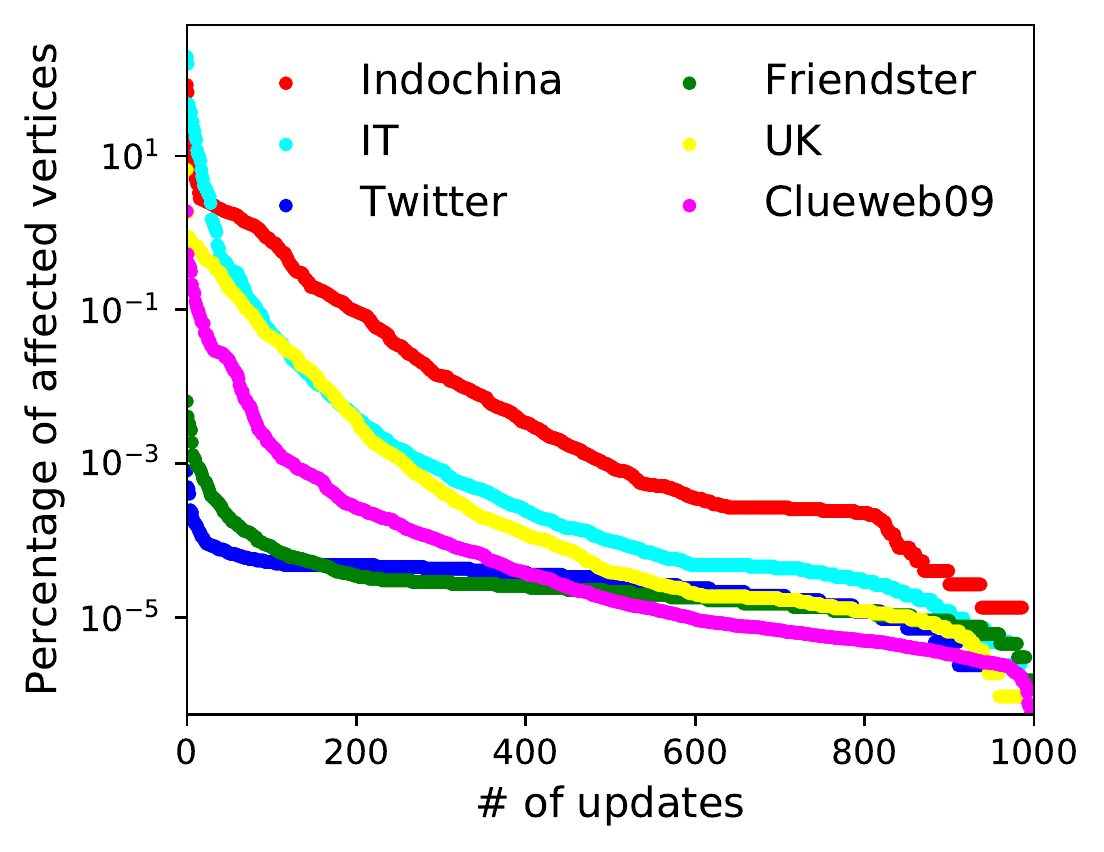}\vspace{-0.2cm}
\caption{Distribution of affected vertices by a single graph change in various networks, where the results for 1000 graph changes are sorted in the descending order.}
\label{fig:percentage_av}\vspace{-0.3cm}
\end{figure}

Previous studies have primarily focused on distance queries on static graphs \cite{akiba2013fast,fu2013label,jin2012highway,abraham2011hub,abraham2012hierarchical,wei2010tedi,farhan2018highly}, with little attention being paid to dynamics on graphs. To speed up query response time, a key technique is to precompute a data structure called \emph{distance labelling} that satisfies certain properties such as 2-hop cover \cite{cohen2003reachability}, and then use this data structure to answer distance queries efficiently. However, when a graph dynamically changes, its distance labelling needs to be changed accordingly; otherwise, distance queries may yield overestimated distances. Although it is possible to recompute a distance labelling from scratch, this leads to inefficiency. As shown in Figure \ref{fig:percentage_av}, the percentage of affected vertices by a single change often ranges from $10^{-5}\%$ to $10\%$ in various real-world networks, 
recomputing distance labelling from scratch for each single change not only wastes computing resources, but also may generate inaccurate query results during recomputing process. The question arising is thus how to efficiently and accurately change distance labelling on dynamic graphs in order to support distance queries?

In this paper, we aim to develop an \emph{online incremental} method that can dynamically maintain distance labelling on graphs being changed by edge and vertex insertions. Typically, real-world dynamic networks are more vulnerable to insertions than removals and a plethora of such real-world networks are large and frequently updated, primarily accommodating insertions \cite{leskovec2007graph,viswanath2009evolution}. Thus, an online incremental method for dynamic graphs should possess the following desirable characteristics: (1) \emph{time efficiency} - It can answer distance queries and update distance labelling efficiently (in the order of milliseconds); 2) \emph{space efficiency} - It guarantees the minimum size of distance labelling to reduce storage costs; (3) \emph{scalability} - It can scale to very large networks with billions of vertices and edges. 

\vspace{0.2cm}
\noindent{\textbf{Challenges.~}}Designing online incremental methods for distance queries on dynamic graphs is known to be challenging \cite{akiba2014dynamic}. When an edge or a vertex is inserted into a graph, outdated and redundant entries of distance labelling may occur. It was reported that removing such entries is a complicated task \cite{akiba2014dynamic} because affected vertices need to be precisely identified so as to update their labels without violating the original properties of a distance labelling such as minimality. Further, although query time and update time are both critical for answering distance queries on dynamic graphs, it is not easy (if not impossible) to design a solution that is efficient in both. This requires us to find new insights into dynamic properties of a distance labelling, as well as a good trade-off between query time and update time.
Last but not least, scaling distance queries to dynamic graphs with billions of nodes and edges is hard. Previous work \cite{akiba2014dynamic,hayashi2016fully} mostly considered 2-hop labelling, which has very high space requirements and index construction time; as a result, their query and update performance are dramatically degraded on large-scale dynamic graphs. Ideally, the labelling size of a graph should be much smaller than its original size. However, the state-of-the-art distance labelling technique, i.e. pruned landmark labeling method (PLL) \cite{akiba2014dynamic}, still yields a distance labelling whose size is 20-30 times larger than the original size of a dataset.

\vspace{0.2cm}
\noindent{\textbf{Contributions.}} Our contributions are summarised as follows:

\vspace{-0cm}
\begin{itemize}
\item Our method overcomes the challenge of eliminating outdated and redundant distance entries. None of the previous studies have addressed this challenge because detecting those entries is too costly \cite{akiba2014dynamic,d2019fully}. When an edge or a vertex is inserted, previous studies only add new distance entries or modify existing distance entries. This would however lead to an ever increasing size of labelling, particularly when a graph is frequently updated by newly added edges or vertices. Accordingly, both query performance and space efficiency would deteriorate over time. 

\item We prove the correctness of our proposed method and show that it preserves the desirable property of minimality on our distance labelling. Due to a property called highway cover \cite{farhan2018highly}, the minimal size of a distance labelling in this work is much smaller than the size of a 2-hop labelling in previous work \cite{akiba2014dynamic,hayashi2016fully}. Preserving minimality on a distance labelling thus improves space efficiency and query performance, as well as update performance. We also provide a complexity analysis of our proposed method.

\item We conducted experiments using 12 real-world large networks across different domains to show the efficiency, scalability and robustness of our method. Particularly, our method can perform updates under one second, on average, even on billion-scale networks, while still answering queries efficiently in the order of milliseconds and guaranteeing the labelling size of a graph to be much smaller. 
\end{itemize}


%% file: section_RelatedWork.tex
\section{RELATED WORK}\label{section:background}

Answering shortest-path distance queries in graphs has been an active research topic for many years. Traditionally, a distance query can be answered using Dijkstra's algorithm \cite{tarjan1983data} on positively weighted graphs or Breadth-First Search (BFS) algorithm on unweighted graphs. However, these traditional algorithms fail to achieve desired response time for distance queries on large graphs. Later, labelling-based methods have emerged as an attractive way of accelerating response time to distance queries \cite{cohen2003reachability,akiba2013fast,jin2012highway,fu2013label,abraham2012hierarchical,abraham2011hub,farhan2018highly}, among which Akiba et al. \cite{akiba2013fast} proposed a pruned landmark labeling (PLL) to precompute a 2-hop cover distance labelling \cite{cohen2003reachability}. 
This method serves as the state-of-the-art for labelling-based distance queries and can handle graphs with hundreds of millions of edges.

So far, only a few attempts have been made to study distance queries over dynamic graphs \cite{akiba2014dynamic,hayashi2016fully}, which are all based on the idea of 2-hop distance labelling or its variants. Akiba et al. \cite{akiba2014dynamic} studied the problem of updating a pruned landmark labelling for incremental updates (i.e. vertex additions and edge additions). This work however does not remove redundant entries in distance labels because the authors considered that detecting such outdated entries is too costly. This inevitably breaks the minimality of pruned landmark labelling, leading to an ever increase of labelling size and deteriorated query performance over time. To accelerate shortest-path distance queries on large networks, another line of research is to combine a partial distance labelling with online shortest-path searches. Hayashi et al. \cite{hayashi2016fully} proposed a fully dynamic approach that selects a small set of landmarks $R$ and precompute a shortest-path tree (SPT) rooted at each $r \in R$. Then, an online search is conducted on a sparsified graph under an upper distance bound being computed via the SPTs. Nevertheless, this method still fails to construct labelling on networks with billions of vertices. Following the same line, a recent work by Farhan et al. \cite{farhan2018highly} introduced a highway-cover labelling method (HL), which can provide fast response time (milliseconds) for distance queries even on billion-scale graphs. However, this approach only works for static graphs. 

%% file: section_Preliminaries.tex
\section{Problem Formulation}
Let $G = (V, E)$ be an undirected graph where $V$ is a set of vertices and $E$ is a set of edges. We denote by $N(v)$ the set of neighbors of a vertex $v \in V$, i.e. $N(v) = \{u \in V | (u, v) \in E \}$. Given two  vertices $u$ and $v$ in $G$, the \emph{distance} between $u$ and $v$, denoted as $d_G(u, v)$, is the length of the shortest path from $u$ to $v$. If there does not exist a path from $u$ to $v$, then $d_G(u, v) = \infty$. We use $P_{G}(u,v)$ to denote the set of all shortest paths between $u$ and $v$ in $G$.
Given a graph $G = (V, E)$, an \emph{edge insertion} is to add an edge $(a, b)$ into $G$ where $\{a,b\}\subseteq V$ and $(a, b)\notin E$. Accordingly, a \emph{node insertion} is to add a new node into $G$ together with a set of edge insertions that connect $v$ to existing vertices in $G$. 
The following fact is critical for designing algorithms for an edge insertion. 

\begin{fact}\label{fact-insertion}
Let $G'=(V,E\cup\{(u,v)\})$ be the graph after inserting an edge $(u, v)$ into $G=(V,E)$. Then for any two vertices $s,t\in V$, $d_G(s, t) \geq d_{G'}(s, t)$.
\end{fact}

That is, the distance between any two vertices never increases after inserting edges or vertices in a graph. 

\vspace{0.15cm}
\noindent\textbf{Highway cover labelling.~}Unlike the previous work \cite{akiba2014dynamic,hayashi2016fully, d2019fully} that uses 2-hop cover labelling \cite{cohen2003reachability}, we develop our method using a highly scalable labelling approach, called \emph{highway cover labelling} \cite{farhan2018highly}.
Let $R\subseteq V$ be a small set of \emph{landmarks} in a graph $G=(V,E)$. For each vertex $v \in V$, the \emph{label} of $v$ is a set of \emph{distance entries} $L(v)= \{(r_1, \delta_L(r_1, v)), \dots, (r_n, \delta_L(r_n,v))\}$, where $r_i \in R$ and $\delta_L(r_i, v) = d_G(r_i, v)$. We call $L = \{L(v)\}_{v \in V}$ a \emph{distance labelling} over $G$ whose \emph{size} is defined as: $size(L)=\sum_{v\in V}|L(v)|$. A \emph{highway} $H=(R, \delta_H)$ consists of a set $R$ of landmarks and a distance decoding function $\delta_H : R \times R \rightarrow \mathbb{N}^+$ such that, for any two landmarks $r_1, r_2\ \in R$,  $\delta_H(r_1, r_2) = d_G(r_1, r_2)$ holds.

\begin{definition}\label{def:highway-cover}
A \emph{highway cover labelling} is a pair $\Gamma=(H, L)$ where $H$ is a highway and $L$ is a distance labelling s.t. for any vertex $v \in V\backslash R$ and $r\in R$, we have:
\begin{align}
d_G(r, v) = \texttt{min}\{\delta_L(r_i, v) + \delta_H(r, r_i) | (r_i, \delta_L(r_i, v)) \in  L(v)\}.
\end{align}
\end{definition}

Highway cover labelling enjoys several nice theoretical properties, such as minimality and order independence. A minimal highway cover labelling can be efficiently constructed, independently of the order of applying landmarks \cite{farhan2018highly}.

Given a highway cover labeling $\Gamma=(H, L)$, an upper bound on the distance between any two vertices $u, v \in V \backslash R$ is computed:
\begin{align}\label{eq:upper-distance-bound}
d^{\top}_{uv} = \texttt{min}\{\delta_L(r_i, u) + \delta_H(r_i, r_j) + \delta_L(r_j, v) |\notag\hspace{1.3cm}\\(r_i, \delta_L(r_i, u)) \in  L(u), (r_j, \delta_L(r_j, v)) \in  L(v)\} 
\end{align}

An exact distance query $Q(u,v,\Gamma)$ can be answered by conducting a distance-bounded shortest-path search over a sparsified graph $G[V \backslash R]$ (i.e., removing all landmarks in $R$ from $G$) under the upper bound $d^{\top}_{uv}$ such that:
\[
    Q(u,v,\Gamma)= 
\begin{cases}
    d_{G[V \backslash R]}(u,v)& \text{if } d_{G[V \backslash R]}(u,v)\leq d^{\top}_{uv},\\
    d^{\top}_{uv}              & \text{otherwise}.
\end{cases}
\]

\vspace{0.15cm}
\noindent\textbf{Problem definition.~}In this work, we study the problem of answering distance queries over a graph that is dynamically changed by edge and vertex insertions over time. Since a vertex insertion can be treated as a set of edge insertions, without loss of generality, below we define the problem based on edge insertions.
\begin{definition}\label{prob:realtime}
Let $G\hookrightarrow G'$ denote that a graph $G$ is changed to a graph $G'$ by an edge insertion. The \emph{dynamic distance querying} problem is, given any two vertices $u$ and $v$ in the changed graph $G'$, to efficiently compute the distance $d_{G'}(u,v)$.
\end{definition}

%% file: section_IncrementalUpdates.tex
\begin{figure*}[ht]
\centering
\includegraphics[width=0.9\textwidth]{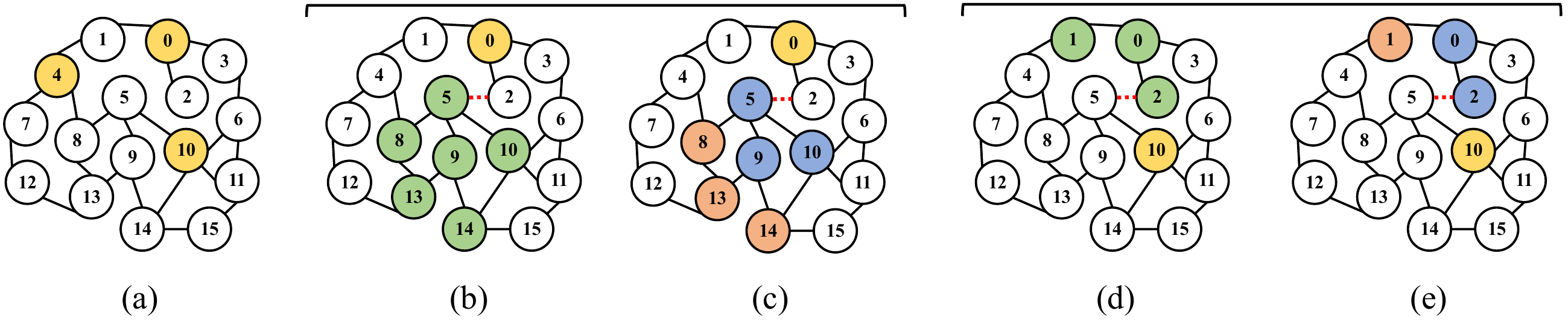}\vspace{-0.5cm}
\caption{An illustration of our online incremental algorithm $\textsc{IncHL}^{+}$: (a) a graph with three landmarks $0$, $4$ and $10$ (colored in yellow); (b) and (d) the BFSs for finding affected vertices (colored in green) w.r.t. landmarks $0$ and $10$, respectively; (c) and (e) the BFSs for repairing affected vertices w.r.t. landmarks $0$ and $10$, respectively, where vertices with added/modified entries are colored in blue, and vertices with removed entries are colored in red.}
\label{fig:incremental-algorithm}\vspace{-0.2cm}
\end{figure*}

\section{Online Incremental Algorithm}\label{subsection:addition}
In this section, we propose an algorithm $\textsc{IncHL}^{+}$ to incrementally update labelling to reflect graph changes. Algorithm \ref{algo:inc-algo} describes the main steps of $\textsc{IncHL}^{+}$. Below, we discuss them in detail.

\subsection{Finding Affected Vertices}
When an update operation occurs on a graph $G=(V,E)$, there exists a subset of ``affected'' vertices in $V$ whose labels need to be updated as a consequence of this update operation on the graph. 

\begin{definition}\label{def:affected_vertex}
A vertex $v\in V$ is \emph{affected} by $G \hookrightarrow G'$ iff $P_{G}(v, r)\not=P_{G'}(v,r)$ for at least one $r\in R$; \emph{unaffected} otherwise.
\end{definition}

We use $\Lambda_r$ to denote the set of all affected vertices w.r.t. a landmark $r$ and $\Lambda=\bigcup_{r\in R}\Lambda_r$ the set of all affected vertices. 

\begin{example}
Consider Figure~\ref{fig:incremental-algorithm}(a) in which $0$ and $10$ are two landmarks. After inserting an edge $(2,5)$, $\Lambda_0=\{5, 8, 9, 10, 13, 14\}$ in Figure~\ref{fig:incremental-algorithm}(b) and $\Lambda_{10}=\{0, 1, 2\}$ in Figure~\ref{fig:incremental-algorithm}(d). 
\end{example}

The following lemma states how affected vertices relate to an edge being inserted.
\begin{lemma}\label{lemma:aff_vs_insertion}
When $G \hookrightarrow G'$ for an edge insertion $(a,b)$, a vertex $v\in \Lambda_r$ iff there exists a shortest path between $v$ and $r$ in $G'$ passing through $(a, b)$. 
\end{lemma}

Following Lemma \ref{lemma:aff_vs_insertion}, we can reduce the search space of affected vertices by eliminating landmarks $r$ with $d_{G}(r,a) = d_{G}(r,b)$ since $\Lambda_r=\emptyset$ in such a case. Thus, we assume that $d_G(r, b) > d_G(r, a)$ w.r.t. a landmark $r$ in the rest of this section w.l.o.g.
Further, by the lemma below, we can also reduce the search space by ``jumping'' from the root of a BFS to vertex $b$.

\begin{lemma}\label{traverse-identifying-affected}
When $G \hookrightarrow G'$ with an inserted edge $(a,b)$, we have $d_G(v,r) \geq d_G(a,r)+1$ for any affected vertex $v\in \Lambda_r$.
\end{lemma}
\begin{proof} 
By Lemma \ref{lemma:aff_vs_insertion}, there exists a shortest path from any affected vertex $v$ to $r$ going through the edge $(a, b)$ and thus through $a$. Since $a$ is unaffected and the distance from $a$ to $v$ is equal to or greater than 1, $d_G(v, r) \geq d_G(a, r) + 1$ thus holds.
\end{proof}
 
Algorithm \ref{algo:algo-affected} describes our algorithm for finding affected vertices. Given a graph $G$ with an inserted edge $(a, b)$ and a highway cover labelling $\Gamma=(H,L)$ over $G$, we conduct a \emph{jumped} BFS w.r.t. a landmark $r$ starting from the vertex $b$ with its new depth $\pi = Q(r, a, \Gamma) + 1$ (Lines 3-4). For every $(v, \pi) \in \mathcal{Q}$, we enqueue all the neighbors of $v$ that are affected into $\mathcal{Q}$ with new distances $\pi + 1$ (Lines 7-8) and add $v$ to $\Lambda_r$ as affected vertex (Line 9). This process  continues until $\mathcal{Q}$ is empty.

\begin{example}
Figure~\ref{fig:incremental-algorithm} illustrates how our algorithm finds affected vertices as a result of inserting an edge $(2,5)$. The BFS rooted at landmark $0$ is depicted in Figure \ref{fig:incremental-algorithm}(b), which jumps to vertex $5$ and finds six affected vertices $\{5, 8, 9, 10, 13, 14\}$. Similarly, the BFS rooted at landmark $10$ is depicted in Figure \ref{fig:incremental-algorithm}(d), which jumps to vertex $2$ and finds three affected vertices $\{0, 1, 2\}$.
\end{example}

\begin{algorithm}[!t]
\caption{Incremental algorithm ($\textsc{IncHL}^{+}$).}
\label{algo:inc-algo}
\KwIn{$G$, $G'$, $(a, b)$, $\Gamma=(H,L)$}
\KwOut{$\Gamma'=(H',L')$}

\ForEach{$r \in R$}{
    $\Lambda_{r} \gets \textsc{FindAffected}(G, (a, b), r, \Gamma)$ \\
    $\textsc{RepairAffected}(G', (a, b), \Lambda_{r}, r, \Gamma)$
}
\end{algorithm}

\begin{algorithm}[!t]
\caption{Finding affected vertices.}
\label{algo:algo-affected}
\SetKwFunction{FMain}{\textsc{FindAffected}}
\SetKwProg{Fn}{Function}{}{end}
\Fn{\FMain{$G$, $(a, b)$, $r$, $\Gamma$}}{

    $\mathcal{Q} \gets \emptyset$, $\Lambda_{r} \gets \emptyset$ \\
    $\pi \gets Q(r, a, \Gamma)+1$ \\
    Enqueue $(b, \pi)$ to $\mathcal{Q}$ \\
    \While{$\mathcal{Q}$ is not empty}{
        Dequeue $(v, \pi)$ from $\mathcal{Q}$ \\
        
        \ForEach{$w \in N(v)$ s.t. $Q(r, w, \Gamma) \geq \pi+1$}{
            Enqueue $(w, \pi+1)$ to $\mathcal{Q}$
        }
        $\Lambda_{r} = \Lambda_{r} \cup \{v\}$
    }
    
    \textbf{return} $\Lambda_{r}$
}
\end{algorithm}

\subsection{Repairing Affected Vertices}
Now we propose a repair strategy to efficiently update the labels of affected vertices in order to reflect graph changes. The key idea is that, instead of conducting a full BFS on all vertices, we conduct a partial BFS from $b$ only on affected vertices. Further, to avoid unnecessary computations, we distinguish two kinds of affected vertices: (1) affected vertices that are \emph{covered} by other landmarks and can thus be easily repaired by removing an entry from their labels; (2) affected vertices whose labels need to be repaired with accurately calculated distances on a changed graph. The following lemma characterizes the first kind according to the definition of highway cover labelling.
\begin{lemma}\label{lem:prunable-vertex}
An affected vertex $v\in \Lambda_r$ is \emph{covered} by a landmark $r'\in R\backslash\{r\}$ iff $r'$ exists in $P_{G'}(v,r)$. If an affected vertex $v\in \Lambda_r$ is covered by $r'$, then any affected vertex $v'\in \Lambda_r$ satisfying $d_{G'}(r,v')=d_{G'}(r,v)+d_{G'}(v,v')$ must also be covered by $r'$.
\end{lemma}


By Lemma \ref{lem:prunable-vertex}, we can efficiently repair affected vertices $v \in \Lambda_r$ as follows. If $v$ is covered by a landmark $r'\in R\backslash\{r\}$ (i.e., one of the unaffected parents of $v$ does not contain $r$ in its label) and is also a landmark, we only update the highway; otherwise, we remove the entry of $r$ from $L(v)$. If $v$ is not covered by any $r'\in R\backslash\{r\}$, we add/modify the entry of $r$ in $L(v)$. If $v$ is a descendant of covered vertices, we simply remove the entry of $r$ from $L(v)$ (if exists).



\begin{algorithm}[t] 
\caption{Repairing affected vertices.}
\label{algo:algo-repair}
\SetKwFunction{FMain}{\textsc{RepairAffected}}
\SetKwProg{Fn}{Function}{}{end}
\Fn{\FMain{$G'$, $(a, b)$, $\Lambda_r$, $r$, $\Gamma$}}{

    $\mathcal{Q}_{uncovered} \gets \emptyset$, $\mathcal{Q}_{covered} \gets \emptyset$ \\
    $\pi \gets d_{G}(r, a)+1$ \\
    Enqueue $(b, \pi)$ to $\mathcal{Q}_{covered}$ if covered; otherwise to $\mathcal{Q}_{uncovered}$ \\
    \While{$\mathcal{Q}_{uncovered}$ is not empty}{
      \While{$(v, \pi) \in \mathcal{Q}_{uncovered}$ \text{at depth} $\pi$}{
        \ForAll{$w \in N(v)$ s.t. $w \in \Lambda_r$ at depth $\pi + 1$}{
          \uIf{$covered(w, \pi + 1)$}{
            \uIf{$w$ is a landmark}{
              $\delta_H(r, w) \gets \pi + 1$
            }\Else{
              Remove $r$ from $L(w)$
            }
            Enqueue $(w, \pi + 1)$ to $\mathcal{Q}_{covered}$
          }\Else{
            Add/Modify $\{(r,\pi + 1)\}$ in $L(w)$ \\
            Enqueue $(w, \pi + 1)$ to $\mathcal{Q}_{uncovered}$
          }
          Remove $w$ from $\Lambda_r$
        }
        Dequeue $(v, \pi)$ from $\mathcal{Q}_{uncovered}$
      }
      
      \While{$(v, \pi) \in \mathcal{Q}_{covered}$ at depth $\pi$}{
        \ForAll{$w \in N(v)$ s.t. $w \in \Lambda_r$ at depth $\pi + 1$}{
          Remove $r$ from $L(w)$ \\
          Remove $w$ from $\Lambda_r$ \\
          Enqueue $(w, \pi + 1)$ to $\mathcal{Q}_{covered}$
        }
        Dequeue $(v, \pi)$ from $\mathcal{Q}_{covered}$
      }
    }
    Remove entry $r$ from remaining vertices in $\mathcal{Q}_{covered}$
}
\end{algorithm}

Algorithm \ref{algo:algo-repair} describes our algorithm for repairing affected vertices. Given a graph $G$ with an inserted edge $(a, b)$ and a set of affected vertices $\Lambda_r$, we conduct a BFS w.r.t. a landmark $r$ starting from the vertex $b$ with its new distance $\pi = d_G(r, a)+ 1$ (Lines 3-4). We use two queues $\mathcal{Q}_{uncovered}$ and $\mathcal{Q}_{covered}$ to process uncovered and covered vertices, respectively. If $b$ is covered, we enqueue $(b, \pi)$ to $\mathcal{Q}_{covered}$ and remove the entry of $r$ from the labels of affected vertices (Line 25). Otherwise, we enqueue $(b, \pi)$ to $\mathcal{Q}_{uncovered}$ and start processing vertices in $\mathcal{Q}_{uncovered}$ (Line 5). For each vertex $v \in \mathcal{Q}_{uncovered}$ at depth $\pi$, we examine its affected neighbors $w$ at depth $\pi + 1$. If $w$ is covered, then if $w$ is a landmark, we update the highway (Line 10); otherwise we remove the entry of $r$ from $L(w)$ (Line 12) because there must exist another landmark in the shortest path from $w$ to $r$ and add $(w, \pi + 1)$ to $\mathcal{Q}_{covered}$ (Line 13). Otherwise, we add/modify the entry of $r$ with the new distance $\pi + 1$ in $L(w)$ and enqueue $w$ to $\mathcal{Q}_{uncovered}$ (Lines 15-16). After that, we remove $w$ from $\Lambda_r$ (line 17). Then, for each $(v, \pi) \in \mathcal{Q}_{covered}$, we remove $r$ from the labels of affected neighbors of $v$, remove these affected vertices from $\Lambda_r$ and enqueue them to $\mathcal{Q}_{covered}$ (Lines 19-24). We process these two queues, one after the other, until $\mathcal{Q}_{uncovered}$ is empty. Finally, we remove the entry of $r$ from the labels of the remaining vertices in $\mathcal{Q}_{covered}$ (Line 25).

\begin{example}
Figure~\ref{fig:incremental-algorithm} illustrates how our algorithm repairs labels as a result of inserting an edge $(2,5)$. The BFS for landmark $0$ is depicted in Figure \ref{fig:incremental-algorithm}(c), which jumps to vertex $5$ and repairs three affected vertices $\{5, 9, 10\}$. The vertices $\{8, 13, 14\}$ are covered by landmarks $4$ and $10$. Similarly, the BFS for landmark $10$ is depicted in Figure \ref{fig:incremental-algorithm}(e), in which vertices $\{0, 2\}$ are repaired and vertex $1$ is covered by landmarks $0$ and $4$.
\end{example}

%% file: section_TheoreticalResults.tex
\section{Theoretical Results}\label{sec:theory}
\begin{table*}[ht!]
 \centering
 \caption{Comparing the update time, query time and labelling size of our method with the baseline methods.}
 \label{table:performance}\vspace{-0.4cm}
 \begin{tabular}{| l || r r r | r r r | r r r |}  \hline
	\multirow{2}{*}{Dataset \hspace*{1cm}}&\multicolumn{3}{c|}{{Update Time} (ms)}&\multicolumn{3}{c|}{Query Time (ms)}&\multicolumn{3}{c|}{Labelling Size} \\\cline{2-10}
    & \hspace{0.2cm}$\textsc{IncHL}^{+}$ & \hspace{0.2cm}\textsc{IncFD} & \hspace{0.2cm}\textsc{IncPLL} & \hspace{0.2cm}$\textsc{IncHL}^{+}$ & \hspace{0.2cm}\textsc{IncFD} & \hspace{0.2cm}\textsc{IncPLL} & \hspace{0.3cm}$\textsc{IncHL}^{+}$ & \hspace{0.3cm}\textsc{IncFD} & \hspace{0.3cm}\textsc{IncPLL} \\

   \hline\hline
	Skitter & 0.194 & 0.444 & 2.05 & 0.027 & 0.019 & 0.047 & 42 MB & 153 MB & 2.44 GB \\
    Flickr & 0.006 & 0.074 & 1.73 & 0.007 & 0.012 & 0.064 & 34 MB & 152 MB & 3.69 GB \\
    Hollywood & 0.031 & 0.101 & 48 & 0.027 & 0.037 & 0.109 & 27 MB & 263 MB & 12.58 GB \\
    Orkut & 2.026 & 2.049 & - & 0.101 & 0.103 & - & 70 MB & 711 MB & - \\
    Enwiki & 0.134 & 0.163 & 5.91 & 0.054 & 0.035 & 0.071 & 82 MB & 608 MB & 12.57 GB \\
    Livejournal & 0.245 & 0.268 & - & 0.044 & 0.046 & - & 122 MB & 663 MB & - \\ \hline
    Indochina & 5.443 & 158 & 2018 & 0.737 & 0.839 & 0.063 & 81 MB & 838 MB & 18.64 GB \\
    IT & 95.92 & 224 & - & 1.069 & 1.013 & - & 854 MB & 4.74 GB & - \\
    Twitter & 0.027 & 0.134 & - & 0.863 & 0.177 & - & 1.14 GB & 3.83 GB & - \\
    Friendster & 0.159 & 0.419 & - & 0.814 & 0.904 & - & 2.43 GB & 9.14 GB & - \\
    UK & 11.49 & 384 & - & 3.443 & 5.858 & - & 1.78 GB & 11.8 GB & - \\
    Clueweb09 & 40.68 & - & - & 16.93 & - & - & 163 GB & - & - \\
   \hline
 \end{tabular}\vspace{-0.2cm}
\end{table*}
\begin{table}[h!]
\centering
\caption{Summary of datasets.}
\label{table:datasets}\vspace{-0.4cm}
\begin{tabular}{| l l | r r r r|} 
  \hline
  Dataset & Network & $|V|$ & $|E|$ & avg. deg & avg. dist \\
  \hline\hline
  Skitter & comp (u) & 1.7M & 11M & 13.081 & 5.1 \\
  Flickr & social (u) & 1.7M & 16M & 18.133 & 5.3 \\
  Hollywood & social (u) & 1.1M & 114M & 98.913 & 3.9 \\
  Orkut & social (u) & 3.1M & 117M & 76.281 & 4.2 \\
  Enwiki & social (d) & 4.2M & 101M & 43.746 & 3.4 \\
  Livejournal & social (d)& 4.8M & 69M & 17.679 & 5.6 \\ \hline
  Indochina & web (d)& 7.4M & 194M & 40.725 & 7.7 \\
  IT & web (d) & 41M & 1.2B & 49.768 & 7.0 \\
  Twitter & social (d) & 42M & 1.5B & 57.741 & 3.6 \\
  Friendster & social (u) & 66M & 1.8B & 55.056 & 5.0 \\ 
  UK & web (d) & 106M & 3.7B & 62.772 & 6.9 \\
  Clueweb09 & web (d) & 1.7B & 7.8B & 9.27 & 7.4  \\
  \hline
\end{tabular}\vspace{-0.2cm}
\end{table}

\textbf{Proof of correctness.~}For $G\hookrightarrow G'$ where our method $\textsc{IncHL}^{+}$ updates a highway cover labelling $\Gamma$ over $G$ into a highway cover labelling $\Gamma'$ over $G'$, we consider $\textsc{IncHL}^{+}$ to be \emph{correct} iff, whenever $Q(u,v, \Gamma)=d_{G}(u,v)$ holds for any two vertices $u$ and $v$ in $G$,  then $Q(u',v', \Gamma')=d_{G'}(u',v')$ also holds for any two vertices $u'$ and $v'$ in $G'$. We prove the theorem below for $\textsc{IncHL}^{+}$. 
\begin{theorem}\label{the:correctness}
$\textsc{IncHL}^{+}$ is correct. 
\end{theorem}\vspace*{-0.2cm}
\begin{proof} 
First, we prove that \texttt{FindAffected} returns the set of all affected vertices $\Lambda_r$ as a result of an edge insertion. $\textsc{IncHL}^{+}$ (Lines 7-8 of Algorithm \ref{algo:algo-affected}) guarantees that any vertex being added to $\mathcal{Q}$ has one shortest path to a landmark $r$ which goes through the inserted edge $(a,b)$. By Lemma \ref{lemma:aff_vs_insertion}, such vertices are affected vertices, and thus a vertex $v$ is added to $\mathcal{Q}$ in Algorithm \ref{algo:algo-affected} iff $v\in \Lambda_r$. Then, we prove that \texttt{RepairAffected} repairs $\Gamma=(H,L)$ s.t. (1) $(r, d_{G'}(r, v)) \in L(v)$ for $v \in \Lambda_r$, iff $P_{G'}(r,v)$ contains only one landmark $r$; (2) $\delta_H(r,r')=d_{G'}(r,r')$ for any $r'\in R \backslash\{r\}$. Starting from $b$ with new distance $\pi$, the distances of affected vertices in $\Lambda_r$ are iteratively inferred on $G'$ and reflected into their labels via $\mathcal{Q}_{uncovered}$ if these affected vertices are not covered (Lines 15-16 of Algorithm \ref{algo:algo-repair}). If an affected vertex $v$ is covered, it is kept in $\mathcal{Q}_{covered}$; if $v$ is also a landmark, $\delta_H(r,v)$ in $H$ is updated (Lines 9-10). Thus, the distance entry of $r$ is removed from the labels of affected vertices appearing in $\mathcal{Q}_{covered}$, whereas any vertex $v$ appearing in $\mathcal{Q}_{uncovered}$ must have $(r, d_{G'}(r,v))\in L(v)$.
\end{proof}

\noindent\textbf{Preservation of minimality.~}
It has been reported in \cite{farhan2018highly} that, given a graph $G$, a minimal highway cover labelling $\Gamma=(H,L)$ of $G$ can be constructed using an algorithm proposed in their work, i.e., $size(L')\geq size(L)$ holds for any $\Gamma'=(H,L')$ of $G$. For $G\hookrightarrow G'$ where $\textsc{IncHL}^{+}$ updates  $\Gamma$ over $G$ into  $\Gamma'$ over $G'$, we prove that $\textsc{IncHL}^{+}$ preserves the minimality of labelling.

\begin{theorem}
If $\Gamma$ is minimal on $G$, then $\Gamma'$ is minimal on $G'$.
\end{theorem}\vspace*{-0.2cm}
\begin{proof}
By Lemma \ref{lem:prunable-vertex}, $(r, d_{G'}(r, v)) \in L(v)$ for $v\in\Lambda_r$ iff $P_{G'}(r,v)$ does not contain any other landmark $R\backslash\{r\}$; otherwise we remove the entry of $r$ from the label of $v$ (Line 12, 21 and 25 of Algorithm \ref{algo:algo-repair}). Thus, the labels of all affected vertices must be minimal after applying $\textsc{IncHL}^{+}$. For unaffected vertices, their labels should remain unchanged. Hence,  $\Gamma'$ must be minimal.
\end{proof}

\noindent\textbf{Complexity analysis.~}Let $m$ be the total number of affected vertices, $l$ be the average size of labels (i.e. $l=size(L)/|V|$), and $d$ be the average degree of vertices. For a landmark, Algorithm \ref{algo:algo-affected} takes $O(mdl)$ time to find all affected vertices and Algorithm \ref{algo:algo-repair} takes $O(md)$ to repair the labels of all affected vertices. We omit $l$ from $O(md)$ for Algorithm \ref{algo:algo-repair}  because distances for all unaffected neighbors of affected vertices are stored in Algorithm \ref{algo:algo-affected}. Therefore, $\textsc{IncHL}^{+}$ has time complexity $O(|R| \times mdl)$. In our experiments, we notice that $m$ is usually orders of magnitudes smaller than $|V|$ and $l$ is also significantly smaller than $|R|$.

\medskip
\noindent\textbf{Directed and weighted graphs.~}
For directed graphs, we can store sets of forward and backward labels, namely $L_f(v)$ and $L_b(v)$, for each vertex $v$ which contain pairs $(r_i, \delta_{r_iv})$ from forward and backward BFSs w.r.t. each landmark. Accordingly, we can store forward and backward highways $H_f$ and $H_b$. Then, we conduct two BFSs to update these labels and highways: one in the forward direction and the other in the backward direction. Our method can also be easily extended to handling weighted graphs by using Dijkstra’s algorithm instead of BFSs.

%% file: section_Experiments.tex
\section{EXPERIMENTS}\label{section:experiments}
We have evaluated our method to answer the following questions: (Q1) How efficiently can our method perform against state-of-the-art methods? (Q2) How does the number of landmarks affect the performance of our method? (Q3) How does our method scale to perform updates occurring rapidly in large dynamic networks? 

\smallskip
\noindent\textbf{Datasets.}~We used 12 large real-world networks as detailed in Table \ref{table:datasets}. These networks are accessible at Stanford Network Analysis Project \cite{leskovec2015snap}, Laboratory for web Algorithmics \cite{BoVWFI}, Koblenz Network Collection \cite{kunegis2013konect}, and Network Repository \cite{rossi2015network}. We treated these networks as undirected and unweighted graphs. 

\smallskip
\noindent\textbf{Updates and queries.}~For each network, we randomly sampled 1,000 pairs of vertices as edge insertions, denoted as $E_I$, where $E_I \cap E = \emptyset$ to evaluate the average update time. Further, we evaluate the average query time with 100,000 randomly sampled pairs of vertices from each network and report the labelling size after reflecting all the updates.

\smallskip
\noindent\textbf{Baseline methods.}~We compared our method ($\textsc{IncHL}^{+}$) with the state-of-the-art methods: (1) \textsc{IncPLL}: an online incremental algorithm proposed in \cite{akiba2014dynamic} which is based on the 2-hop cover labelling to answer distance queries; (2) \textsc{IncFD}: an online incremental algorithm proposed in \cite{hayashi2016fully} which combines a 2-hop cover labelling with a graph traversal algorithm to answer distance queries. 
The codes of these methods were provided by their authors and implemented in C++. We used the same parameter settings for these methods as suggested by their authors unless otherwise stated. For a fair comparison, following \cite{hayashi2016fully} we set $|R|=20$ for \textsc{IncFD} and our methods, except for Clueweb09 which has $|R|=150$ due to its billion-scale vertices. Our methods were implemented in C++11 and compiled using gcc 5.5.0 with the -O3 option. We performed all the experiments using a single thread on Linux server (Intel Xeon W-2175 with 2.50GHz and 512GB of main memory). 
\subsection{Performance Comparison}

\begin{figure*}[t!]
\centering
\includegraphics[width=0.9\textwidth]{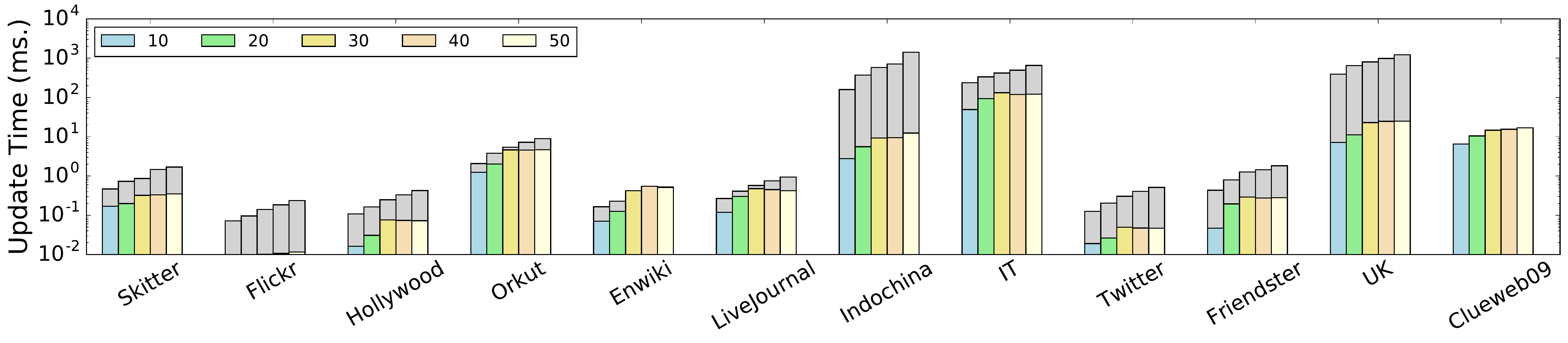}\vspace{-0.6cm}
\caption{Average update time of our method $\textsc{IncHL}^{+}$ (in colored bars) and the baseline method \textsc{IncFD} (in colored plus grey bars) under 10-50 landmarks. There are no results of \textsc{IncFD} for Clueweb09 due to the scalability issue.}
\label{fig:varying-landmarks}\vspace{-0.3cm}
\end{figure*}
\begin{figure*}[t!]
\includegraphics[width=0.9\textwidth]{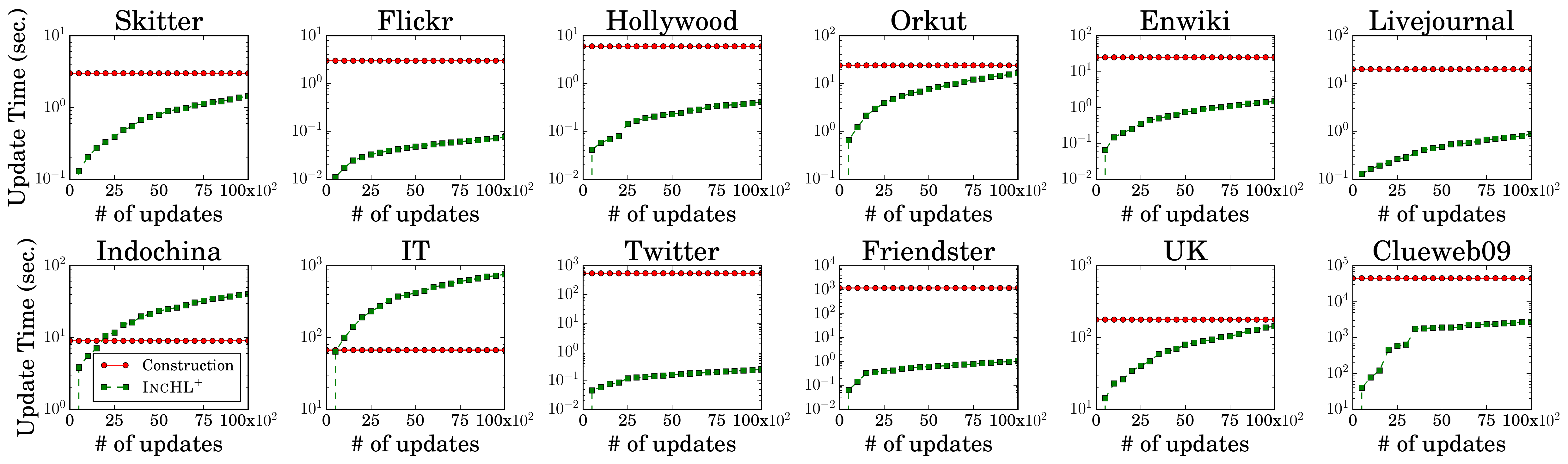}\vspace{-0.45cm}
\caption{Update time of $\textsc{IncHL}^{+}$ for performing up to 10,000 updates against construction time of labelling from scratch.} 
\label{fig:aut_batchsize_inc}\vspace{-0.35cm}
\end{figure*}

\subsubsection{Update Time}
Table \ref{table:performance} shows that the average update time of our method $\textsc{IncHL}^{+}$ outperforms the state-of-the-art methods \textsc{IncFD} and \textsc{IncPLL} on all datasets. This is due to a novel repair strategy utilized by $\textsc{IncHL}^{+}$. Further, only $\textsc{IncHL}^{+}$ can scale to very large networks with billions of vertices and edges. \textsc{IncFD} fails to scale to Clueweb09, and \textsc{IncPLL} fails for 7 out of 12 datasets due to very high preprocessing time and memory requirements.

\subsubsection{Labelling Size} \label{subsection:ls}
From Table \ref{table:performance}, we see that $\textsc{IncHL}^{+}$ has significantly smaller labelling sizes than \textsc{IncFD} and \textsc{IncPLL}. When updates occur on a graph, the labelling sizes of \textsc{IncFD} and $\textsc{IncHL}^{+}$ remain stable because their average label sizes are bounded by the size of landmarks set (i.e. $|R|$). Moreover, \textsc{IncFD} stores complete shortest path trees w.r.t. landmarks; while $\textsc{IncHL}^{+}$ stores pruned shortest-path trees which lead to labelling of much smaller sizes than IncFD. For \textsc{IncPLL}, the labelling sizes may increase because \textsc{IncPLL} does not remove outdated and redundant entries.

\subsubsection{Query Time}
In Table \ref{table:performance} the query times of $\textsc{IncHL}^{+}$ are comparable with \textsc{IncFD} and \textsc{IncPLL}. It has been shown in \cite{d2019fully} that query time depends on labelling size. As discussed in Section \ref{subsection:ls}, the update operations do not considerably affect the labelling sizes of \textsc{IncFD} and $\textsc{IncHL}^{+}$, and thus their query times remain stable. However, the query times for \textsc{IncPLL} may increase over time because of the presence of outdated and redundant entries, which result in labelling of increasing size.

\subsection{Performance with Varying Landmarks}
Figure \ref{fig:varying-landmarks} shows the average update time of our method $\textsc{IncHL}^{+}$ against the baseline method \textsc{IncFD} under varying landmarks, i.e., $|R| \in [10, 20, 30, $ $40, 50]$. As we can see, $\textsc{IncHL}^{+}$ outperforms \textsc{IncFD} on all the datasets against almost every selection of landmarks. We can also see the performance gap remains stable for most of the datasets when increasing the number of landmarks. This empirically verifies the efficiency of our repair strategy.

\subsection{Scalability Test}
We conducted a scalability test on the update time of our method $\textsc{IncHL}^{+}$, by starting with 500 updates and then iteratively adding 500 updates each time until 10,000 updates. Figure \ref{fig:aut_batchsize_inc} shows the results. We observe that the update time of $\textsc{IncHL}^{+}$ on almost all the datasets is considerably below the construction time of labelling. On Indochina and IT, $\textsc{IncHL}^{+}$ performs relatively worse because these networks have large average distances as depicted in Table \ref{table:datasets}, which lead to high percentages of affected vertices as shown in Figure \ref{fig:percentage_av}. In contrast, $\textsc{IncHL}^{+}$ performs well on graphs with small average distances such as Twitter. Overall, $\textsc{IncHL}^{+}$ can scale to perform a large number of updates efficiently.

%% file: section_Conclusion.tex
\section{CONCLUSION}\label{section:conclusion}
This paper has studied the problem of answering distance queries on large dynamic networks. Our proposed algorithm exploits properties of a recent labelling technique called highway cover labelling \cite{farhan2018highly} to efficiently process incremental graph updates, and can preserve the minimality property of labelling after each update operation. We have empirically evaluated the efficiency and scalability of the proposed algorithm. The results show that our proposed algorithm outperforms the state-of-the-art methods. In future, we plan to further investigate the effects of decremental updates on graphs since they are also commonly used in practice. 